\documentclass[aps,pra,9pt,twocolumn,showpacs,superscriptaddress]{revtex4-2}
\usepackage{physics}
\usepackage{graphicx}
\usepackage[colorlinks=true,linkcolor=red,citecolor=blue,plainpages=false,pdfpagelabels]{hyperref}
\usepackage{times}
\usepackage{amssymb}
\usepackage{amsthm}
\newtheorem{theorem}{Theorem}
\newtheorem{definition}{Definition}

\usepackage{mathtools}
\usepackage{amsmath}
\usepackage{csquotes}

\begin{document}
\title{No-masking theorem for observables}

\author{Swapnil Bhowmick}
\author{Abhay Srivastav}\email{abhaysrivastav.hri@gmail.com}
\author{Arun Kumar Pati}

\affiliation{ Quantum Information and Computation Group,\\
Harish-Chandra Research Institute, A CI of Homi Bhabha National Institute, Chhatnag Road, Jhunsi, Prayagraj 211019, India
}

\begin{abstract}
The no-masking theorem for quantum information proves that it is impossible to encode an arbitrary input state into a larger bipartite entangled state such that the full information is stored in the correlation but the individual subsystems have no information about the input state. Here, we ask the question: Is it possible to mask an observable such that the information about the observable is available in the joint system, but individual subsystems reveal nothing about the imprints of the observable? This generalizes the notion of masking to observables. We show that a universal unitary that can mask an arbitrary observable in any dimension does not exist. For a qubit system, we show that the masking operation for a given observable is locally unitarily connected to the SWAP operation. This suggests a conservation law for information content of observables that goes beyond the conservation laws under symmetry operations. Furthermore, we prove that the unconditional no-bit commitment result follows from the no-masking theorem for observables. Our results can have important applications in quantum information and quantum communication where we encode information not in states but in observables.
\end{abstract}

\maketitle


\section{ Introduction} The linear and unitary nature of dynamical evolution in quantum mechanics places several restrictions on the information theoretic tasks that can be performed. These restrictions have been formulated using several no-go theorems such as no-cloning \cite{wootters1982}, no-broadcasting \cite{barnum1996}, no-deleting \cite{akpati}, no-hiding \cite{braunstein} and no-masking \cite{modi2018} theorems. However, most of these theorems are formulated so that the quantum information is encoded in the states. In quantum theory, the states evolve according to the Schr\"odinger equation while the observables evolve according to the Heisenberg equation of motion. However, not much is known about what kind of limitations are imposed on the manipulation of observables when they evolve in time. It is, therefore, natural and fundamental to ask whether we can formulate no-go theorems for observables. In Ref.~\cite{ferraro2006} the no-cloning theorem was formulated for information encoded in the statistics of observables.

The notion of masking of quantum information was introduced in Ref.~\cite{modi2018}. 
Traditionally, information is stored in physical systems which may be classical or quantum. The process of masking is to investigate an alternative way of storing information where it is stored in the quantum correlations among two or more systems, rather than in the systems themselves. This allows one to make quantum information inaccessible to everyone locally. 
The physical intuition of masking is that given a distinguishable set of objects, we encode their information content in non-local correlations in such a way that they become locally indistinguishable. In the case of states, the masker maps a set of states to a set of entangled states whose marginal states are equal for all inputs.
It was proved in Ref.~\cite{modi2018} that an arbitrary pure state of a qubit or qudit cannot be masked in the quantum correlation of a bipartite system. Subsequently, the notion of masking was discussed in the multipartite scenario \cite{li2018,Han_2020,shi2021,shang2021,hu2022}, for some restricted sets of states \cite{liang2019,ding2020,sun2021,du2021,zhu2021}, and also for non-Hermitian systems \cite{Lv_2022}. Experimental realization of the masking of quantum information for some restricted sets of states was done in Refs. \cite{liu2021,zhang2021}. The notion of masking of quantum information by non-unitary processes was discussed in Refs. \cite{ming2019,liang2020,modi2020,lie2020,choi2021,abhay2022}. Recently, the notion of work masking was introduced in Ref.~\cite{parashar} where it was shown that the work content of a state cannot be masked by an energy preserving unitary.

In this paper, we explore if it is possible to mask an arbitrary observable for a quantum system. Given a quantum system it has some physical attributes such as the energy, spin, angular momentum, and so on. Typically, we always imagine that the physical property of a system is akin to the system. What we would like to ask is whether the physical properties of a system can be made blind to the system itself, i.e., under some general transformation we map the observable to identity. When we represent an observable by identity it has nothing in it.
We define the notion of masking of observables and show that masking a set of arbitrary observables is not possible in any dimension $d$. However, we find that it is possible to mask a known observable in $d=2$ by choosing an appropriate unitary. Surprisingly, this is not the case in arbitrary $d$ where we find some observables that cannot be masked. Moreover, we find surprisingly that for a qubit system, the masked observable can be retrieved from the environment itself by an observer who has knowledge of the global evolution. This points to a fundamental notion of conservation of physical observable beyond what is taught in textbooks. 
As an application of our results, we show that the impossibility of unconditional bit commitment (see Refs. \cite{brassard_1993,Lo_1997,Mayers_1997,brassard_1997,bub_2001,ariano_2007,Lie2021quantumonetime,kent_2015} and references therein) can be understood as a necessary consequence of the no-masking theorem for observables. Our results can have deep impact and important applications in thermalization of observables, information scrambling as well as variety of other areas in quantum information and quantum communication.

\section{Preliminaries}
Let $\{{\lambda}_i\}$ be the $(d^2-1)$ generators of $SU(d)$ with the following properties: (i) ${\lambda}_i={\lambda}^\dag_i$, (ii) $\Tr (\lambda_i)=0$ and (iii) $\Tr (\lambda_i \lambda_j)=2\delta_{ij}$.  We can construct these operators using the following prescription, with $\{\lambda_i \}_{i=1}^{d^2-1}= \{u_{jk}, v_{jk}, w_l \}$:
\begin{eqnarray}\label{sudgenerators}
    u_{jk}&=&\ketbra{j}{k}+\ketbra{k}{j}, \ \ v_{jk}= -i(\ketbra{j}{k}-\ketbra{k}{j}), \nonumber\\
    w_l &=& f(l)\qty(\sum_{m=1}^l \ketbra{m}{m} -l\ketbra{l+1}{l+1}),
\end{eqnarray}
where $1\leq j< k \leq d$,  $f(l)=\sqrt{\frac{2}{l(l+1)}}$ \text{and} $1\leq l\leq d-1$. Let ${\cal{H}}^d_X$ be the Hilbert space for a system $X$ with dimension $d$. A state $\rho$ of the system $X$ is a linear operator with the following properties: (i) $\Tr(\rho)=1$, (ii) $\rho^{\dag}=\rho$, (iii) $\rho \geq0$. Let $\mathbb{D}(\mathcal{H}^d_X)$ denote the set of all states of the system $X$. Then the generalized Bloch sphere representation of a state $\rho \in \mathbb{D}(\mathcal{H}^d_X)$ is given as \cite{KIMURA2003,Bertlmann_2008}
\begin{equation}
    \rho=\frac{1}{d}\mathbb{I} +  \vec{\omega}\cdot \boldsymbol{\lambda},\label{state}
\end{equation}
where $\vec{\omega} \in \mathbb{R}^{d^2-1}$ and is the called {\it{Bloch vector}} and $\boldsymbol{\lambda}$ is the vector comprised of $\{\lambda_i \}$. Eq.(\ref{state}) satisfies the first two properties of a state by construction. A necessary condition for positivity of $\rho$ is $|\vec{\omega}|^2 \leq \frac{d-1}{2d}$ which is sufficient for $d=2$. However, in higher dimensions there exist additional conditions on $\vec{\omega}$ to ensure the positivity of $\rho$ \cite{KIMURA2003}.

In the Schr\"odinger picture, states evolve in time while the observables are fixed. The evolution of a quantum state in this picture is given by a linear, completely positive and trace preserving map called quantum channel. Any quantum channel $\mathcal{E}$ acting on a state $\rho$ can be written as $\mathcal{E} (\rho)=\sum_i E_i \rho E_i^\dag$ where $\{E_{i}\}$ are the Kraus operators of the channel such that $\sum_{i}E_{i}^\dag E_{i}=\mathbb{I}$ \cite{CHOI1975285,wilde}. In the Heisenberg picture, however, the states are fixed while the observables evolve in time by the adjoint of the channel $\mathcal{E}$, i.e., $\mathcal{E}^*$. The action of $\mathcal{E}^*$ on an observable $\mathcal{O}$ can be represented as $ \mathcal{E}^* (\mathcal{O})=\sum_{i} E_i^\dag \mathcal{O} E_i$. Note that $\mathcal{E}^*$ is unital because $ \mathcal{E}^* (\mathbb{I})=\sum_{i} E_i^\dag E_i=\mathbb{I}$. The expectation value of an observable calculated in these two pictures is the same as  $\Tr(\mathcal{E}(\rho)\mathcal{O})=\Tr(\rho\mathcal{E}^*(\mathcal{O}))$.

\section{Masking of Observables}
Let us consider a $d$-dimensional system $A$ and its environment $E$ with Hilbert space $\mathcal{H}_A\otimes\mathcal{H}_E$. Let $\mathcal{O}_A$ be an observable of the system $A$ which is by definition a Hermitian operator and can be parametrized as $\mathcal{O}_A = a_0\mathbb{I}+\vec{a}\cdot\boldsymbol{\lambda}$, where $a_0\in\mathbb{R},\vec{a}\in \mathbb{R}^{d^2-1}$. 
Masking of an observable is defined such that the information of the observable is hidden from $A$ but is known globally. By \enquote{information of observable} we mean the degeneracy structure of the observable along with its eigenstates which induces its corresponding measurement setup, i.e., a positive operator valued measure (POVM). Since, this remains unchanged under a transformation of the type $\mathcal{O}\rightarrow \alpha(\mathcal{O}- \beta\mathbb{I}/{d})$. By taking $\alpha= 1/|\vec{a}|$ and $\beta= \Tr(\mathcal{O})$ we can set $a_0 =0$ and $|\vec{a}|=1$. Thus for the purpose of masking, we consider observables of the form $\mathcal{O}_A=\hat{a}\cdot\boldsymbol{\lambda}$. If an observable is fully degenerate, i.e., proportional to identity, it contains no information in the sense that no two states can be distinguished by measuring that observable.
Since masking is a physical process, it can be represented by a unitary $U$ acting on the joint system-environment state $AE$. We now formally define masking of observables.

\begin{definition}
A unitary $U$ acting on $\mathcal{H}_A\otimes\mathcal{H}_E$ is said to mask information of an observable $\mathcal{O}_A$ by mapping it to $\mathcal{O}_{AE}^{'}$ such that the locally evolved observable is fully degenerate,
\begin{equation}
\mathcal{O}'_A=\mathcal{E}^*(\mathcal{O}_{A})=\mathbb{I}, \label{maskcond1}
\end{equation}
Here, $\mathcal{E}^*$ is the local channel acting on the system $A$, corresponding to the unitary $U$.
\end{definition}

The masking operation for states is defined in such a way that the information about the initial state is hidden from both the marginals of the evolved global state. However, in our case the initial observable of the environment $E$ is assumed to be $\mathcal{O}_E=\mathbb{I}$. Now let the local channel corresponding to the masker $U$, that acts on the environment be $\mathcal{F}^*$. Since $\mathcal{F}^*$ is unital, the evolved observable of the environment is $\mathcal{O}'_E=\mathcal{F}^*(\mathbb{I})=\mathbb{I}$, which has no information about $\mathcal{O}_A$. Hence, the above condition is automatically satisfied for the environment. Therefore, in order to define masking of observables it is sufficient to consider the local evolution of only the system observable.

Note that our definition of the masking of observables is not Heisenberg equivalent to that for masking of states in the Schr\"odinger picture.
The masker for states is an isometry $U$ which maps a set $\{\ket{\psi^k}_A\}$ to $\{\ket{\phi^k}_{AB}\}$ such that \cite{modi2018}
\begin{equation}
    \Tr_{A(B)}\ketbra{\phi^{m}}=\Tr_{A(B)}\ketbra{\phi^{n}} ~ \forall~ m, n.
\end{equation}
The crucial difference here is that the information that is masked in the case of states is encoded in the index $k$. However, the information that is masked in the case of observables is contained in their degeneracy structure. Due to this fact, masking can be defined for a single observable, but is only possible for a set of states. \\

We can now rewrite the masking condition in terms of the set of output states of $\mathcal{E}$ as follows. For $\mathcal{O}_A$ to be maskable Eq.(\ref{maskcond1}) implies that $\forall ~ \rho\in\mathbb{D}(\mathcal{H}^{d}_A)$ we have
\begin{eqnarray}
    1 &=& \Tr( \rho\mathcal{E^*}(\mathcal{O}_A)) \nonumber \\ 
    &=& \Tr(\mathcal{E}(\rho)\mathcal{O}_A) \nonumber \\
    &=& \Tr(\vec{b}\cdot\boldsymbol{\lambda}~\hat{a}\cdot\boldsymbol{\lambda}) \nonumber\\
    &\Rightarrow&  \hat{a}\cdot\vec{b}   = \frac{1}{2},\label{maskcond2}
\end{eqnarray}
where $\hat{a}$ and $\vec{b}$ are vectors in $\mathbb{R}^{d^2-1}$ and $\abs{\hat{a}}=1$. The first equality follows by averaging with respect to $\rho$ on both sides of Eq.(\ref{maskcond1}). The second equality is the equivalence of the Schr\"odinger and the Heisenberg picture. The third equality and the rest follow by expanding the observable and the state in the basis $\{\mathbb{I},\boldsymbol{\lambda}\}$ as $\mathcal{O}_A = \hat{a}\cdot \boldsymbol{\lambda}$ and $\mathcal{E}(\rho)= \mathbb{I}/d+ \vec{b}\cdot\boldsymbol{\lambda}$ where $\hat{a}, \vec{b}\in\mathbb{R}^{d^2-1}$ and simplifying thereafter. Eq.({\ref{maskcond2}}) is the masking condition for observables in terms of the set of output Bloch vectors, $S=\{\vec{b} :\mathbb{I}/d+ \vec{b}\cdot\boldsymbol{\lambda} =\mathcal{E}(\rho),\rho\in\mathbb{D}(\mathcal{H}^{d}_A)\}$ and is equivalent to Eq.(\ref{maskcond1}). It says that an observable $\mathcal{O}_A$ can be masked if $\exists$ a set $S$ of output Bloch vectors  such that $\forall~\vec{b}\in S$, Eq.(\ref{maskcond2}) holds.  We will use either Eq.\eqref{maskcond1} or Eq.\eqref{maskcond2} as the masking condition for observables as per convenience.  
Note that any given observable $\mathcal{O}_A$ can be masked in $d=2$ because we can always find a channel $\mathcal{E}$ such that $\vec{b}$ satisfies Eq.(\ref{maskcond2}). Now a natural question arises: Does there exist a universal masker that can mask all observables in $d=2$? We show that this is not possible in the following theorem.
\begin{theorem}
No masker can mask all observables in $d=2$. \label{nomasking2}
\end{theorem}
\renewcommand\qedsymbol{$\blacksquare$}
\begin{proof}
Let there exists a universal masker and $\mathcal{E}^*$ be its local channel mapping any observable $\mathcal{O}_A$ to $\mathbb{I}$. Since the channel $\mathcal{E}^*$ is linear, in order to map all observables of the form $\mathcal{O}_A=\hat{a}\cdot \boldsymbol{\sigma}$  to $\mathbb{I}$, it should be able to map the set $\{\sigma^1,\sigma^2,\sigma^3\}$ to $\mathbb{I}$, where $\sigma^i$ is the $i$-th Pauli matrix. Let us now assume that $\mathcal{O}_A=\sigma^3$, i.e, $\hat{a}=(0,0,1)$, satisfies Eq.(\ref{maskcond1}), then Eq.(\ref{maskcond2}) implies that $b_3=1/2$. Since for $\mathcal{E}(\rho)$ to be positive, $b_1^2+b_2^2+b_3^2\leq b_0^2=1/4$, it then implies that $b_1=b_2=0$, i.e., $\vec{b}=(0,0,1/2)$ and $\mathcal{E}(\rho)=\ketbra{0}~\forall~\rho$. Therefore, any linear combination of $\sigma^1$ and $\sigma^2$, say $a_1\sigma^1+a_2\sigma^2$ cannot satisfy Eq.(\ref{maskcond2}). This proves that if $\sigma^3$ satisfies Eq.(\ref{maskcond1}) for some channel $\mathcal{E}^*$, then  $\sigma^1$ and $\sigma^2$ cannot do so for the same channel. It implies that there does not exist a channel that can map all observables $\mathcal{O}_A$ to $\mathbb{I}$. Therefore, there does not exist a universal masker that can mask all observables in $d=2$.
\end{proof}
Note that in the case of state masking it is always possible to mask a set of distinguishable states by mapping it unitarily to a set of orthogonal entangled states. However, in the case of observables for qubit systems it is impossible to mask even an orthogonal set of observables as shown in the above theorem. In terms of the set of output Bloch vectors  $S$, the above theorem can be rephrased as following: there does not exist a convex subset $S$ of the Bloch sphere such that for each $\mathcal{O}_A$, Eq.(\ref{maskcond2}) holds $\forall$ $\vec{b}\in S$. As discussed before the masking operation hides the information of the observable from the system 
 by mapping it to identity. Since the initial observable of the environment is chosen to be identity and its local evolution is unital, the information of the system observable remains inaccessible to the environment locally. So, a natural question to ask is that where has this information gone and how to retrieve it? We answer this question in the following theorem.
\begin{theorem} \label{infretriev}
For an arbitrary qubit observable $\mathcal{O}_A\in\{\hat{n}\cdot\boldsymbol{\sigma}:\hat{n} \text{ is a unit vector}\}$ under the masking operation, if the information about $\mathcal{O}_A$ is lost from $A$, then it can be retrieved from the environment by an observer who can access the global unitary of system-environment. 
\end{theorem}
\renewcommand\qedsymbol{$\blacksquare$}
\begin{proof}
See Appendix A.
\end{proof}

In physics, symmetry operations lead to conservation laws for physical
systems. For example, if $G$ is the generator of some symmetry operation
and if the latter commutes with the Hamiltonian, then the physical
observable represented by the operator $G$ is conserved. However, since the unitary that implements the masking process for an arbitrary observable $\mathcal{O}_A$ of the type $\hat{n}\cdot\boldsymbol{\sigma}$, does not commute with the observable $\mathcal{O}_A\otimes\mathbb{I}_E$ (see Appendix A for details), the
above theorem suggests a conservation law for observables beyond
the usual symmetry operation.

\section{ Masking in arbitrary dimensions} 
We have seen that any known observable can be masked in $d=2$. We will now show that this is not true for arbitrary dimension $d$. Consider the basis $\{\lambda_i\}_{i=1}^{d^2-1}=\{u_{jk}, v_{jk}, w_l\}$ for operators on $\mathcal{H}^d$ as defined in Eq.\eqref{sudgenerators}. Since $u_{jk}$ and $v_{jk}$ are Hermitian operators with eigenvalues $\{\pm 1,0\}$, they can always be masked by a constant channel. Consider a particular observable from this set $u_{12}$ as an example. It can be masked by a constant channel $\mathcal{E}$ which maps all states $\rho$ to the eigenstate $\ket{\phi}$(say) with $+1$ eigenvalue. This is because $\Tr(\mathcal{E}(\rho)u_{12})=\bra{\phi}u_{12}\ket{\phi}=1~\forall \rho\in \mathbb{D}(\mathcal{H}^{d})$. However, $\{w_l\}$ are not maskable in general. Consider an observable $w_{l}$ from this set. It can be masked if there exists a channel $\mathcal{E}$ such that $\Tr(\mathcal{E}(\rho) w_{l})=1 \ \forall \rho\in \mathbb{D}(\mathcal{H}^{d})$. The expectation value of an operator in any state is upper bounded by its maximum eigenvalue. However, the maximum eigenvalue of $w_{l}$, $\sqrt{\frac{2}{l(l+1)}}<1~~\forall~~l\geq2$. Therefore, the expectation value of $w_{l}$ in any state can never be 1 for $l\geq2$. Hence the set $\{w_l\}$ is not maskable for $l\geq2$.
Let us now define $\mathcal{M}^d$ as the set of all maskable observables in $d$ dimensions. Then we can ask whether there exists a universal unitary that can mask all observables in $\mathcal{M}^d$. In Theorem \ref{nomasking2} we showed that this is not true for $d=2$ which implies that a set of arbitrary observables cannot be masked by a single unitary. We will now show that this result holds for arbitrary dimension.
\begin{theorem} \label{nomaskingd}
A set of arbitrary observables cannot be masked in any dimension $d$ by a single unitary. 
\end{theorem}
\renewcommand\qedsymbol{$\blacksquare$}
\begin{proof}
We will prove the non-existence of a universal masker for arbitrary $d$, using the masking condition \eqref{maskcond2} in terms of the set of output Bloch vectors  $S$. Let $\mathcal{C}_S$ be the set of observables that can be masked by a channel $\mathcal{E}$ with the corresponding set of output Bloch vectors  $S$. Let us assume that $\mathcal{E}$ is the local channel of a universal masker, it is equivalent to saying that there exists a set of output Bloch vectors  $S$ such that $\mathcal{C}_S=\mathcal{M}^d$ is the set of all maskable observables on the Hilbert space $\mathcal{H}^d$. If we now consider a set $S_p=\{\vec{b}\}$ where $\vec{b}$ is some vector in $S$ then it is obvious that $\mathcal{C}_{S}\subseteq\mathcal{C}_{S_p}$. Since $\mathcal{C}_S=\mathcal{M}^d$ and $\mathcal{M}^d$ is the set of all observables, therefore $\mathcal{C}_{S_p}=\mathcal{M}^d$.  Let us now consider a different set $S_{q}=\{\vec{b}'\}$ where $\vec{b}'$ is another vector in the set $S$, then following similar arguments as above we get $\mathcal{C}_{S_{q}}= \mathcal{M}^d$. We will now show that there exists at least one observable $\mathcal{O}=\hat{a}\cdot\boldsymbol{\lambda}\in\mathcal{C}_{S_{q}}$ such that $\mathcal{O}\notin\mathcal{C}_{S_p}$. The masking condition \eqref{maskcond2} gives
    \begin{eqnarray}
       \hat{a}\cdot\vec{b}'-\frac{1}{2}=0.\label{b'masker}
    \end{eqnarray}
   Now the observable $\mathcal{O}$ will also belong to $\mathcal{C}_{S_p}$ if it satisfies the following equation
   \begin{equation}
         \hat{a}\cdot\vec{b}-\frac{1}{2}=0.\label{bmasker}
   \end{equation}
   Taking the difference of the above two equations, we get 
    \begin{eqnarray}
       \hat{a}\cdot(\vec{b} -\vec{b}')=0.
    \end{eqnarray}
   We can always pick a vector $\hat{a}\in\mathbb{R}^{d^2-1}$ such that $\hat{a}\cdot\vec{b}\neq\hat{a}\cdot\vec{b}'$. Specifically, we can always choose $\hat{a}$ that satisfies Eq.\eqref{b'masker} but not Eq.\eqref{bmasker}. Therefore, 
$\mathcal{O}\notin\mathcal{C}_{S_p}$ which implies that 
 $\mathcal{C}_{S_p}\neq\mathcal{M}^d$. Since $\mathcal{C}_{S}\subseteq\mathcal{C}_{S_p}$, therefore $\mathcal{C}_{S}\neq\mathcal{M}^d$. Thus there does not exist a universal masker that can mask all of the observables in the set $\mathcal{M}^d$. This implies that a set of arbitrary observables cannot be masked by a single unitary which completes the proof.
\end{proof}
Note that the physical motivation behind the masking operation allows us to modify \eqref{maskcond1} in such a way that the masking operation maps an observable to $c\mathbb{I}$, $c\in\mathbb{R}$. This will change what observables can be masked and what cannot be masked. However, this will not change the fact that a set of arbitrary observables cannot be masked by a single unitary in any dimension $d$ which is our main physical result. 

Now that we have seen that a single channel cannot mask all observables in $\mathcal{M}^d$ we can ask whether there exists a subset of $\mathcal{M}^d$ that can be masked by a given channel. We will call such a set of observables as {\it comaskable} corresponding to a given channel. From Theorem \ref{nomasking2} and its preceding calculations, we can conclude that in $d=2$ only one observable can be masked by a given channel. This is because in this case, $|\vec{b}|\leq 1/2$ and $\hat{a}\cdot\vec{b}=1/2$ implies that $|\vec{b}|= 1/2$ and $\hat{a}|| \vec{b}$. This fixes the direction of $\hat{a}$ and determines it uniquely. It implies that the set of comaskable observables corresponding to any channel in $d=2$ contains a single observable. However, in higher dimensions, $|\vec{b}|$ can be greater than $1/2$ because the surface of the Bloch ball containing all of the pure states is defined by  $|\vec{b}|^2 = \frac{d-1}{2d}$. This implies that $\hat{a}$ need not be parallel to $\vec{b}$ in order to satisfy the masking condition \eqref{maskcond2}. So, at least in the case of a constant channel whose set of output Bloch vectors  contains a single vector, i.e., $S=\{\vec{b}\}$, there will exist multiple observables that satisfy the masking condition for the same $\vec{b}$. Starting from $\hat{a}$, we can find a comaskable observable $\hat{a}'$ that is infinitesimally close to $\hat{a}$ :
\begin{eqnarray}
    \hat{a}'= \hat{a}+ \epsilon\vec{v},
\end{eqnarray}
where $\vec{v}$ is orthogonal to the plane formed by the linear span of $\{\hat{a},\vec{b}\}$ and $\abs{\epsilon}\ll 1$. We can see that under this infinitesimal transformation the dot product remains unchanged: $\hat{a}\cdot\vec{b}= \hat{a}'\cdot\vec{b}= 1/2$. We can then repeat the above process taking $\hat{a}'$ as the starting vector to find more comaskable observables.
  

\section{ No-bit commitment}
In the quantum bit commitment protocol, Alice commits a bit to Bob, the value of which is revealed to him at a later stage by Alice. The two key properties of any bit commitment protocol are (i) {\it Binding}-- In the revealing phase, Alice cannot change the bit she committed to Bob, and (ii) {\it Concealing}-- Bob cannot identify the bit until she reveals it. The usual protocol runs as follows. Alice encodes her bit $0$ or $1$ in the state $\ket{\Psi^0_{AB}}$ or $\ket{\Psi^1_{AB}}$, respectively, and sends the subsystem $B$ to Bob. Then at a later time she reveals the bit value and sends some information to Bob so that he can verify that Alice was indeed committed to the bit that she later revealed. The (ii) condition implies that the reduced state of $B$ is independent of the bit that Alice commits, i.e., $\rho^0_B=\rho^1_B$. This is equivalent to the following condition $\Tr(\mathcal{O}_B\rho^0_B)=\Tr(\mathcal{O}_B\rho^1_B)~~\forall~\mathcal{O}_B$, i.e., Bob cannot distinguish between $\rho^0_B$ and $\rho^1_B$ by measuring any observable $\mathcal{O}_B$. The (ii) condition also implies that $\ket{\Psi^0_{AB}}=\sum_i \sqrt{\lambda_i}\ket{a^0_i}\ket{b_i}$ and $\ket{\Psi^1_{AB}}=\sum_i \sqrt{\lambda_i}\ket{a^1_i}\ket{b_i}$, where $\rho^0_B=\rho^1_B=\sum_i\lambda_i\ketbra{b_i}$, i.e., the two global states are connected by a local unitary on Alice's system. This is the key to cheating because Alice can change her committed bit by applying a local unitary on her part of the system, thus the (i) condition does not hold. Therefore, the unconditional quantum bit commitment is impossible \cite{Lo_1997,Mayers_1997}. We will now show that the impossibility of quantum bit commitment is a necessary consequence of the no-masking of observables. 

\begin{theorem}
  The no-masking of observables implies the unconditional no-bit commitment. \label{bitmasking}
\end{theorem}
\renewcommand\qedsymbol{$\blacksquare$}
\begin{proof}
  We prove this by showing that the possibility of unconditional quantum bit commitment implies the existence of a universal masker for all observables and then using its contrapositive. We start with a conditional bit commitment protocol which is (i)~imperfectly concealing but (ii)~perfectly binding. (i) means that Bob can only measure a restricted set $\mathcal{Q}$ of observables on his subsystem so that he cannot identify the bit until Alice reveals it, i.e.,  $\Tr(\rho^0_B\mathcal{O}_B)=\Tr(\rho^1_B\mathcal{O}_B)=1~~\forall~\mathcal{O}_B \in \mathcal{Q} $. (ii) implies that there exists $\ket{\Psi^0_{AB}}$ and $\ket{\Psi^1_{AB}}$ such that $(U_A\otimes\mathbb{I}_B)\ket{\Psi^0_{AB}}\neq\ket{\Psi^1_{AB}}~\forall~ U_A$ so that Alice cannot cheat once she has committed a bit to Bob. Now given $\rho^0_B$ and $\rho^1_B$ we can always construct a channel $\mathcal{E}$ such that $\forall~\sigma\in \mathbb{D}(\mathcal{H}^d_B)$
\begin{equation}
    \mathcal{E}(\sigma)=\sum_{i=0,1}\Tr(\Pi^i\sigma)\rho^i_B,\label{universalmasker}
\end{equation}
where $\Pi^0=\ketbra{0}$ and $\Pi^1=\mathbb{I}-\Pi^0$ are the elements of a POVM. The channel measures $\sigma$ in the POVM $\{\Pi^i\}$ and prepares $\rho^i_B$ for the $i$-th outcome. The Kraus operators of $\mathcal{E}$ are $E_{ij}=\sqrt{p^i_j}\ketbra{e^{i}_j}{i}$, where $\rho^i_B=\sum_j p^i_j \ketbra{e^i_j}$ and $\bra{e^i_j}\ket{k}\neq\delta_{jk}$. Now averaging both sides of Eq.(\ref{universalmasker}) with respect to $\mathcal{O}_B$ we get
\begin{eqnarray}
  \Tr(\mathcal{E}(\sigma)\mathcal{O}_B)&=&\Tr[\{\Tr(\Pi^0\sigma)\rho^0_B+\Tr(\Pi^1\sigma)\rho^1_B\}\mathcal{O}_B]\nonumber\\
  &=&\Tr(\Pi^0\sigma)\Tr(\rho^0_B\mathcal{O}_B)+\Tr(\Pi^1\sigma)\Tr(\rho^1_B\mathcal{O}_B)\nonumber\\
  &=&\Tr(\rho^0_B\mathcal{O}_B).
\end{eqnarray}
 Thus, $\Tr(\mathcal{E}(\sigma)\mathcal{O}_B)=1 ~~\forall~\mathcal{O}_B \in \mathcal{Q}$ and $\forall~\sigma\in\mathbb{D}(\mathcal{H}^d_B)$. Using the making condition (\ref{maskcond2}) and the discussion in the preceding section, we can conclude that $\mathcal{Q}$ is the set of comaskable observables corresponding to the channel $\mathcal{E}$. Therefore, for the unconditional bit commitment protocol, $\mathcal{Q}$ is the set of all observables. This implies that the channel $\mathcal{E}$ can mask all observables which contradicts Theorem \ref{nomaskingd}. Hence, the no-masking of observables implies the unconditional no-bit commitment.
\end{proof}

\section{ Conclusions}
We used the concept of degeneracy structure to define the information content of observables. So a fully degenerate observable contains no information. Then we defined the notion of masking of an observable which means that the information of the observable is hidden from both the system and the environment locally but is preserved globally. We show that it is impossible to mask a set of arbitrary observables in any dimension, even though a known observable can always be masked for a qubit system. This shows that physical attributes such as energy, spin, angular momentum, charge, and so on, cannot be made blind to the system itself, where the physical observable evolves under a general transformation and gets mapped to identity. However, it remains to be seen whether the approximate masking of a set of observables by a single unitary is possible, where the observables are mapped infinitesimally close to identity. Furthermore, for the case of qubit we find that if the observable is mapped to identity, then it simply gets swapped to the ancillary system up to a local unitary. This shows that after the masking operation the information of an observable can be retrieved from the environment itself by an observer who has knowledge of the global evolution which can be interpreted as conservation of physical observable in composite systems. Thus, physical attributes of a quantum system can never be lost. If it is lost from one subsystem, then it may be found in another subsystem which is nothing but a conservation law beyond symmetry transformations. An extension of this work would be to see if the same result holds in $d$ dimensional systems as well. We also introduced the notion of comaskable observables as a set of observables that can be masked by a given channel. We provide a prescription to construct comaskable observables in $d>2$ for a constant channel. It would be interesting to see if these sets can be fully characterized for a given channel in arbitrary dimensions. As an application of our result we prove that the famous no-bit commitment is a necessary consequence of the no-masking theorem for observables. These results can have deep impact on how we manipulate observables and what kind of limitations may be imposed on them. Finally, we believe that our results will have important applications in thermalization of observables, information scrambling as well as variety of other areas in quantum information and quantum communication in future.

\appendix
\section{Proof of theorem 2}
\renewcommand\qedsymbol{$\blacksquare$}
\begin{proof}
Let us first consider the masking operation for $\sigma^3$. There exists a channel $\mathcal{E}^*$ such that $\mathcal{E^*}(\sigma^3)=\mathbb{I}$, i.e., $\sum_{i}E_{i}^{\dagger}\sigma^3 E_{i}=\sum_i\ketbra{i}{i}$. It can be easily seen that $E_i=\ketbra{0}{i}$ for $i=\{0,1\}$ are the Kraus operators of the channel. We now construct the unitary $U$ that acts on the global observable $\mathcal{O}_{AE}=\sigma^3\otimes\mathbb{I}$ and whose local channel $\mathcal{E}^*$ maps $\sigma^3$ to identity. Since $\mathcal{E}^*$ and $\mathcal{E}$ both correspond to the same unitary, we work with $\mathcal{E}$ for simplicity. To construct the unitary $U$ we will first isometrically extend the channel $\mathcal{E}$ \cite{stinespring1955, wilde}. Any quantum channel can be isometrically extended using its Kraus operators $\{E_i\}$ as follows
\begin{eqnarray}
V_{A\rightarrow{AE}}\equiv\sum_{i}E_{i}\otimes v\ket{i}_E,
\end{eqnarray}
where $v$ is a local isometry acting on the environment. Since, the Kraus operators of the channel in our case are $E_i=\ketbra{0}{i}$, the isometric extension of $\mathcal{E}$ is given by
\begin{eqnarray}
V=\sum_{i}\ketbra{0}{i}\otimes v\ket{i}_E.\label{isometry}
\end{eqnarray}
Now we know that an isometry is part of a unitary on a larger system. Therefore, to get this unitary we have to define its action on the joint system-environment, i.e., $AE$ state. Using Eq.(\ref{isometry}) we see that the isometry $V$ acts on the state $\ket{j}_A$ of the system $A$ as
\begin{eqnarray}
V\ket{j}_A=\ket{0}_A\otimes v\ket{j}_E.
\end{eqnarray}
We now assume that the initial state of the environment is $\ket{0}_E$ and then define the action of $U$ in this case as
\begin{eqnarray}
U(\ket{j}_A\otimes\ket{0}_E)=\ket{0}_A\otimes u_0\ket{j}_E, \label{unidef1}
\end{eqnarray}
where we used the local unitary freedom $u_0$ in defining the action of $U$. To fully specify the unitary we have to show how it evolves the joint system-environment state when the initial state of the environment is $\ket{1}_E$. We choose the evolution in such a way that the overall interaction is unitary
\begin{eqnarray}
U(\ket{j}_A\otimes\ket{1}_E)=\ket{1}_A\otimes u_1\ket{j}_E,\label{unidef2}
\end{eqnarray}
where $u_1$ is a unitary acting on the environment. We can now write the full unitary using Eqs.(\ref{unidef1}) and (\ref{unidef2}) as follows
\begin{eqnarray}
U=\sum_{ij}\ketbra{j}{i}\otimes u_j\ketbra{i}{j},
\label{unitaryappendix}
\end{eqnarray}
which acts as the swap operator for $\sigma^3\otimes\mathbb{I}$.\\
Now consider an arbitrary observable $\mathcal{O}_A= \hat{a}\cdot \boldsymbol{\sigma}$ for some unit vector $\hat{a}$. Let us now fix a coordinate system $(\hat{x},\hat{y},\hat{z})$. There exists a proper rotation $R$ such that $\hat{a}= R \hat{z}$. This rotation $R$ corresponds to a unitary $w$ acting on $\hat{z}\cdot\boldsymbol{\sigma}=\sigma^3$, so $\mathcal{O}_A$ can be written as 
\begin{equation}
    \mathcal{O}_A=w^{\dag}\sigma^3 w.
\end{equation}
Then the masking condition implies that for some channel $\mathcal{E'}^*$
\begin{eqnarray}
    \mathcal{E'}^*(\mathcal{O}_A)&=& \sum_i E'^{\dag}_i(w^{\dag}\sigma^3 w)E'_i
    =\sum_iE_i^\dag\sigma^3 E_i=\mathbb{I},
\end{eqnarray}
where $E_i=w E'_i$ are the Kraus operators of $\mathcal{E}^*$. Therefore, the Kraus operators of the channel $\mathcal{E'}^*$ are $E'_i= w^\dag E_i$. It can be easily seen now that the unitary extension of the channel $\mathcal{E'}^*$ is
\begin{equation}
    U'=(w^\dag \otimes \mathbb{I})U.
\end{equation}
We will now argue that the masker $U'$ obtained above for $\mathcal{O}_A$ is unique up to a local unitary on the environment. The fact that $U'$ is unique for $\mathcal{O}_A$ up to a local unitary on the environment means that the channel $\mathcal{E'}^*$ that maps $\mathcal{O}_A$ to $\mathbb{I}$ is unique for $\mathcal{O}_A$. To see the uniqueness of $\mathcal{E'}^*$ consider its adjoint that acts on states as follows
\begin{eqnarray}
    \mathcal{E'}(\rho)&=&\sum_i E'_i\rho E'^{\dag}_i =\sum_i w^{\dag}E_i\rho E^{\dag}_i w,
\end{eqnarray}
where $E_i$ are the Kraus operators of any channel $\mathcal{E}$ whose adjoint maps $\sigma^3$ to $\mathbb{I}$, and are not assumed to be of the form $E_i=\ketbra{0}{i}$.
Since, we know from the proof of Theorem 1 that $\sum_i E_i\rho E^{\dag}_i=\mathcal{E}(\rho)=\ketbra{0}~\forall~\rho$. Therefore, $\mathcal{E'}(\rho)=w^{\dag}\ketbra{0}w~\forall~\rho$, i.e., $\mathcal{E'}$ is a constant channel and hence is unique. Therefore, its adjoint channel $\mathcal{E'}^*$ is also unique which in turn implies that the masker $U'$ for $\mathcal{O}_A$ is unique up to a local unitary on the environment.\\
Now the action of this masker on the joint system-environment observable is
 \begin{equation}
     U'^{\dag}(\mathcal{O}_A \otimes \mathbb{I}) U' = \mathbb{I} \otimes \sigma^3=\mathbb{I}\otimes w\mathcal{O}_A w^\dagger.
 \end{equation}
The initial observable $\mathcal{O}_A$ can be retrieved from the environment by the action of a local unitary $w$ on it. Note that knowing $w$ alone is not sufficient to retrieve the initial observable from the environment, we need to know the global unitary $U$ as well. This proves our theorem.
\end{proof}

\bibliography{ref.bib}

\end{document}